\documentclass[a4paper,11pt]{article}

\usepackage{amsthm,amsmath,enumerate,color}
\usepackage{amssymb,graphicx}
\usepackage{accents}
\usepackage{mathdots}
\usepackage[tableposition=above]{caption}
\usepackage{subcaption,cite}
\usepackage[title]{appendix}
\usepackage{enumitem}

\usepackage{float}
\usepackage{nicefrac}
\usepackage[utf8]{inputenc}
\usepackage{centernot}
\usepackage{multirow}
\usepackage{multicol}
\usepackage{etoolbox}
\usepackage{lineno}
\usepackage{mathtools}
\usepackage{comment}
\usepackage[normalem]{ulem}
\usepackage[ruled,vlined,linesnumbered]{algorithm2e}
\usepackage[defaultlines=5,all]{nowidow}

\providetoggle{long}
\settoggle{long}{true}
\usepackage{tikz}
\usetikzlibrary{arrows}
\usetikzlibrary{shapes, shapes.geometric}
\tikzstyle{every path}=[thick]

\newcommand{\NULL}{\textsc{null}}

\long\def\Omit#1{}
\def\diss{{\delta}}

\usepackage{float}

\usepackage{tabularx, environ}

\usepackage{bbm,hyperref}
\hypersetup{
	colorlinks=true,
    linkcolor={red!50!black},
    citecolor={blue!50!black},
    urlcolor={blue!80!black},
	bookmarksopen=true,
	bookmarksnumbered,
	bookmarksopenlevel=2,
	bookmarksdepth=3
}
\usepackage[capitalize]{cleveref}
\crefalias{algocf}{algorithm}%
\crefalias{algocfline}{line}%
\crefalias{AlgoLine}{line}%
\makeatletter

\newtheorem{theorem}{Theorem}
\newtheorem{lemma}[theorem]{Lemma}
\newtheorem{observation}[theorem]{Observation}
\newtheorem{corollary}[theorem]{Corollary}
\newtheorem{definition}[theorem]{Definition}
\newtheorem{example}{Example}

\newtheorem{claim}{Claim}
\usepackage{thmtools,thm-restate}

\newcounter{claimCount}
\AtBeginEnvironment{proof}{\setcounter{claimCount}{0}}

\theoremstyle{remark}
\newtheorem*{ccproof}{Proof}

\newenvironment{cproof}[1][\proofname]{%
  \begin{ccproof}%
}{%
\phantom{}~\hfill~$\triangle$
  \end{ccproof}%
}

\newcolumntype{\expand}{}
\long\@namedef{NC@rewrite@\string\expand}{\expandafter\NC@find}

\NewEnviron{fancyproblem}[2][]{%
	\def\problem@arg{#1}%
	\def\problem@framed{framed}%
	\def\problem@lined{lined}%
	\def\problem@doublelined{doublelined}%
	\ifx\problem@arg\@empty%
	\def\problem@hline{}%
	\else%
	\ifx\problem@arg\problem@doublelined%
	\def\problem@hline{\hline\hline}%
	\else%
	\def\problem@hline{\hline}%
	\fi%
	\fi%
	\ifx\problem@arg\problem@framed%
	\def\problem@tablelayout{|>{\bfseries}lX|c}%
	\def\problem@title{\multicolumn{2}{|l|}{%
			\raisebox{-\fboxsep}{\textsc{\large #2}}%
	}}%
	\else
	\def\problem@tablelayout{>{\bfseries}lXc}%
	\def\problem@title{\multicolumn{2}{l}{%
			\raisebox{-\fboxsep}{\textsc{\large #2}}%
	}}%
	\fi%
	\bigskip\par\noindent%
	\begin{tabularx}{\textwidth}{\expand\problem@tablelayout}%
		\problem@hline%
		\problem@title\\[2\fboxsep]%
		\BODY\\\problem@hline%
	\end{tabularx}%
	\medskip\par%
}
\makeatother
\colorlet{darkgreen}{green!60!black}
 	
\def\real{\hbox{\rm\vrule\kern-1pt R}}
\def\nat{\hbox{\rm\vrule\kern-1pt N}}

\newcommand{\PP}{\textsf{P}}
\newcommand{\NP}{\textsf{NP}}

\newcommand{\MINSUM}{\textsc{CG Min-Sum Dissatisfaction}}
\newcommand{\pred}{\textnormal{pred}}
\newcommand{\succc}{\textnormal{succ}}

\newcommand{\probdef}[6][Question]{
\hbox{\vbox{
  \medskip
  \ifthenelse{\equal{#5}{}}{}{\label{#6}}
  \noindent\ifthenelse{\equal{#4}{}}{}{\textsc{#4}\ifthenelse{\equal{#5}{}}{}{ \sloppy\mbox{(\textsc{#5})}}}
  \par
  \smallskip
  \noindent\begin{tabularx}{\textwidth}{@{}l X}
    \textbf{Input:}		& {#2}\\
    \textbf{#1:}		& {#3}
  \end{tabularx}
  \medskip

}}
}

\usepackage{authblk}
\usepackage{microtype}

\title{Allocation of Indivisible Items with a Common Preference Graph: Minimizing Total Dissatisfaction}

\author[1,2]{Nina Chiarelli}
\author[3]{Cl\'ement Dallard}
\author[4]{Andreas Darmann}
\author[4]{Stefan Lendl}
\author[1,2]{Martin Milani\v c}
\author[1]{Peter Mur\v si\v c}
\author[4]{Ulrich Pferschy}

\affil[1]{FAMNIT, University of Primorska, Glagolja\v ska 8, 6000 Koper, Slovenia}
\affil[2]{IAM, University of Primorska, Muzejski trg 2, 6000 Koper, Slovenia}
\affil[3]{Department of Informatics, University of Fribourg, Boulevard de P\'erolles 90, 1700 Fribourg, Switzerland}
\affil[4]{Department of Operations and Information Systems, University of Graz, Universit\"atsstra\ss{}e, 8010 Graz, Austria}

\date{}

\begin{document}
\maketitle

\begin{abstract}
Allocating indivisible items among a set of agents is a frequently studied discrete optimization problem.
In the setting considered in this work, the agents' preferences over the items are assumed to be identical. We consider a very recent measure for the overall quality of an allocation which does not rely on numerical valuations of the items.
Instead, it captures the agents' opinion by a directed acyclic preference graph with vertices representing items.
An arc $(a,b)$ in such a graph means that the agents prefer item $a$ over item $b$.
For a given allocation of items the dissatisfaction of an agent is defined as the number of items which the agent does not receive and for which no more preferred item is given to the agent.
Our goal is to find an efficient allocation of the items to the agents such that the total dissatisfaction over all agents is minimized.

We explore the dichotomy between \NP-hard and polynomially solvable instances, depending on properties of the underlying preference graph.
While the problem is \NP-hard for three agents even on very restricted graph classes, it is polynomially solvable for two agents on general preference graphs.
For an arbitrary number of agents, we derive polynomial-time algorithms for relevant restrictions of the underlying undirected graph.
These are trees and, among the graphs of treewidth two, series-parallel graphs and cactus graphs.


\end{abstract}

\section{Introduction}

The effective allocation of tasks, goods or general objects (from now on called \emph{items}) to the members of a group (called \emph{agents}) is a classical topic in Applied Mathematics.
Often these items represent indivisible elements that can only be assigned as a whole, and each item can be assigned to at most one agent.
This property puts the problem into the area of Discrete Optimization where it gave rise to a large body of literature (e.g.,  \cite{bouveret-survey, darmann2015, lang2016, vetschera2010}).

A wide range of models have been put forward to allocate items to agents on the basis of their preferences over the items, capturing agents' preferences in different ways. These include, e.g., ranking all subsets of items (see~\cite{puppe}), numerical evaluations of the items (see, e.g., \cite{amanatidis2023,bouveret-survey,santaclaus,ourarx}), and ordinal preferences over the items (see, e.g., \cite{aziz,baumeister2017,brams-two-envyfree}).
From a practical point of view,  however, ranking all subsets of items already becomes very tedious for a small number of items; also, it can be quite demanding to come up with a numerical value (representing utility or profit) for each item, if an acceptable monetary valuation is not obvious.
Indeed, for many people it can be very challenging to find meaningful numerical values to express their preferences on a larger set of items.
On the other hand, it seems much easier to compare two items and say which one is liked better.
Also, from a theoretical point of view ordinal valuations have certain advantages over cardinal valuations (see, e.g.,~\cite{aziz}).
However, one may find many pairs of items merely incomparable in several scenarios: the items may  include very different categories of objects; items may consist of several attributes, which do not outweigh each other in an obvious way (compare~\cite{annonini2009});
in presence of choice overload (see, e.g., \cite{buturak, scheibehenne}), one may simply not be able (or it is too challenging) to compare all pairs of alternatives.

Note that incomparability between two items does not imply that these items induce equal value for an agent. This becomes particularly visible when considering items of different types, e.g., bicycles and laptops: one might not be able to express a preference between a certain bicycle and a dedicated laptop, but that does not imply that these are of equal value for the agent --- the items are just too different to be compared. Actually, comparing  items of different types bears the risk of reducing their value to a single metric, such as monetary value, which  may fail to capture the broader significance and qualitatively different values of the items (such as, e.g.,  mobility, exercise and sustainability for bicycles, and access to knowledge and communication for laptops; compare, for instance,~\cite{Sen2004, Walasek2024}).  Therefore, a \textit{partial order} can be seen as a very natural way for agents to express their preferences over a set of heterogenous items.
In this paper we consider the case of identical preferences expressed by a common partial order that is represented by a directed acyclic graph.
Such a group consensus might be reached by agents exchanging their opinions, in particular for sets of less well-known items, or if the agents form a homogenous group with very similar views. In addition, identical preferences may also result from some underlying objective quality of the items that all agents recognize.
Allocations of items to agents with identical preferences over items were treated in the literature, e.g., in \cite{bouveretlang2008, brams2000,freeman}.

For preferences expressed by a partial order the key question remains: How do we evaluate the outcome of an allocation? What is a suitable measure for the utility obtained by an agent from the subset of allocated items?
Recently, in~\cite{general} we introduced a new measure for expressing \emph{satisfaction} or \emph{dissatisfaction} over partial preferences of an agent as follows.
If a certain item is assigned to the agent, then all more preferred items cause dissatisfaction to the agent, unless the agent receives an even more preferred item.
On the other hand, less preferred items contribute to the satisfaction of the agent.

\begin{example}\label{ex:illustrative}
As an illustrative real-world example consider the allocation of toys in a kindergarten.
At the start of the play time each kid is assigned one toy in order to learn how to focus on a single object.
Clearly, some of the toys are more attractive than others, but for many pairs of toys the kids find it hard to choose the preferred toy.
Considering group pressure and adaptation to opinion leaders, it is quite natural that the preferences of all kids tend towards a common consensus.

All kids agree that playing with the tablet is the most coveted (but rarely granted) game.
It is followed by three incomparable, but highly attractive toys:
a remote-controlled 4WD toy car with siren (toy (a)),
a plastic doll with real hair and AI-controlled dialog ability (toy (b)),
and a large brick set of a city scenario (toy (c)).
Each of these three highly ranked toys represents the top choice of a whole category of toys.
Slightly less attractive than (a) is a silent remote-controlled 4WD toy car (d), followed by a remote-controlled sedan (e) and two incomparable choices, namely a hand-operated tow car with a towing crank (f), and a hand-operated fire-engine with a ladder (g).
In the doll box there are two more, but less attractive dolls, incomparable to each other, one which can speak three built-in sentences (h), the other who can water a diaper (i).
If the large city brick set is not available, there is still a choice between a brick castle (j), which is more interesting than a set of knight figures (k), or a unicorn (l) followed by a princess (m).

The preferences over these toys form an out-tree as depicted in Figure~\ref{fig:toy} (formal definitions of graph classes will follow in Section~\ref{sec:prelim}).
If a kid is assigned a certain toy, it will not care too much about the less preferred toys.
However, most kids strive to be the rulers of the play group and would eagerly grab all the more preferred toys as well as all incomparable toys.
Thus, they will voice a dissatisfaction measured by the number of all these toys that they did not receive, and over which they do not prefer the toy they did receive.

\begin{figure}
    \centering
    \begin{tikzpicture}[every node/.style={draw=none,circle,inner sep=0.8pt,minimum size=2mm}]
\node[draw=none] (A1) at (3,6) {tablet};
\node (B11) at (0,4.7) {(a)};
\node (B12) at (3,4.7) {(b)};
\node (B13) at (6,4.7) {(c)};
\node (C1) at (0,3.4) {(d)};
\node (C2) at (2,3.4) {(h)};
\node (C3) at (4,3.4) {(i)};

\node (C4) at (5,3.4) {(j)};
\node (C5) at (7,3.4) {(l)};
\node (D1) at (0,2.1) {(e)};
\node (D2) at (5,2.1) {(k)};
\node (D3) at (7,2.1) {(m)};

\node (E1) at (-1,0.8) {(f)};
\node (E2) at (1,0.8) {(g)};
\draw[->] (A1) -- (B12);
\draw[->] (A1) -- (B13);
\draw[->] (A1) -- (B11);
\draw[->] (B11) -- (C1);
\draw[->] (B12) -- (C2);
\draw[->] (B12) -- (C3);
\draw[->] (B13) -- (C4);
\draw[->] (B13) -- (C5);
\draw[->] (C1) -- (D1);
\draw[->] (C4) -- (D2);
\draw[->] (C5) -- (D3);
\draw[->] (D1) -- (E1);
\draw[->] (D1) -- (E2);
\end{tikzpicture}
    \caption{An out-tree representing the preference graph over toys as described in Example~\ref{ex:illustrative}.}
    \label{fig:toy}
\end{figure}
\end{example}

More formally, the preferences of the agents can be represented by a common acyclic directed \emph{preference graph} $G=(V,A)$, where $V$ is the set of items and an arc $(a,b) \in A$ means that all agents prefer item $a$ over item $b$.
Since preferences should be transitive, two arcs $(a,b)$ and $(b,c)$ imply that $a$ is also preferred over $c$, even if the arc $(a,c)$ is not contained in $A$.
An item $a$ \emph{dominates} item $b$ if there is a directed path (possibly of length $0$) from $a$ to $b$ in $G$; intuitively, an item $a$ dominating an item $b$ expresses that  $a$ is considered ``at least as good as'' $b$, i.e., that all agents prefer $a$ over $b$ or  $a=b$.
Allocating a subset of items $S\subseteq V$ to an agent incurs a satisfaction value equal to the number of items dominated by any item in $S$.
On the other hand the resulting dissatisfaction value is given by the number of items neither allocated to that agent nor dominated by any item in $S$.
Clearly, for any allocation, satisfaction and dissatisfaction values add up to $|V|$ for every agent, since every item contributes to exactly one of the two values in any allocation.

\begin{example}\label{ex:diss}
As an example highlighting the concept of dissatisfaction consider three agents $a_1$, $a_2$, $a_3$ and the preference graph with $8$ items $1,2,\ldots,8$ depicted in Figure~\ref{fig:polytree}.
The graph is a polytree, i.e., ignoring the arc directions translates the graph into an undirected tree. Consider the allocation given by assigning the items $2$, $4$, and $5$ to agent $a_1$ (rectangular vertices), the items $6$ and $7$ to agent $a_2$ (circles), and the items $1$, $3$, and $8$ to agent $a_3$ (diamonds).
For agent $a_1$, this yields a dissatisfaction of $2$, because $a_1$ does not receive items $1$ and $3$.
Note that, e.g., item $8$ does not contribute to the dissatisfaction of the agent, because item $2$ (which the agent receives) dominates item $8$.
Analogously, with the considered allocation the dissatisfaction of agent $a_2$ is $5$ (induced by the items $1$, $2$, $3$, $4$, $5$) and the dissatisfaction of $a_3$ is $4$ (induced by the items $2$, $4$, $5$, $7$).
Observe that $a_3$ would have the same dissatisfaction when receiving only the single item $1$, because items $3$ and $8$ are already dominated by item $1$.
\end{example}

\begin{figure}
    \centering
    \begin{tikzpicture}[every node/.style={draw,circle,inner sep=0.8pt,minimum size=2mm}]
\node (A1)[rectangle,inner sep=2pt] at (3,6) {$4$};
\node (B11) at (1.5,4.7) {$6$};
\node (B12) at (4.5,4.7) {$7$};
\node (C1)[rectangle, inner sep=2pt] at (0,6) {$2$};
\node (C2)[diamond] at (1.5,6) {$3$};
\node (C3)[diamond] at (1.5,3.4) {$8$};

\node (C4)[rectangle, inner sep=2pt] at (4.5,6) {$5$};
\node[diamond] (D1) at (1.5,7.3) {$1$};

\draw[->] (A1) -- (B12);
\draw[->] (A1) -- (B11);
\draw[->] (C1) -- (B11);
\draw[->] (C2) -- (B11);
\draw[->] (B11) -- (C3);
\draw[->] (C4) -- (B12);
\draw[->] (D1) -- (C2);
\end{tikzpicture}
    \caption{A  preference graph (polytree) for an instance with 8 items and three agents $a_1$, $a_2$, $a_3$.
    Items $2$, $4$, and $5$  (rectangular vertices) are assigned to $a_1$ (resulting in a dissatisfaction of $2$), items $6$ and $7$ (circles) are assigned to $a_2$ (for a dissatisfaction of $5$), and items $1$, $3$, and $8$ (diamonds) are assigned to $a_3$ (for a dissatisfaction of $4$).}
    \label{fig:polytree}
\end{figure}

Observe that dissatisfaction is related to the concept of \emph{envy-freeness}, a standard concept in fair division (see, e.g.,~\cite{foley, brams-two-envyfree, bouveret-survey}).
Envy refers to situations in which an agent is willing to exchange its whole bundle of items with the items allocated to another agent, as the agent would prefer the latter set of items to its original allocation.
While there are various options for the definition of such a preference relation for sets of totally ordered items, see~\cite{Barbera2004}, the case of partial orders enforces a quite restrictive definition:
A set of items $S_1$ is preferred over $S_2$ if for each item $s\in S_2$ there exists an item $s' \in S_1$ which is preferred over $s$ according to the partial order.
Clearly, for this definition a set $S_1$ preferred over $S_2$ will also give a lower dissatisfaction value.
However, the notion of envy implied by this preference relation is overly restrictive and can be relaxed by the additive dissatisfaction concept applied in this work.
An agent may envy another agent for some of the items they did not receive and for which they did not receive a dominating item (in Example~\ref{ex:diss}, e.g., agent $a_2$ envies (i) agent $a_3$ for the items $1$ and $3$, and (ii) agent $a_1$ for the items $2$, $4$ and $5$).
This is reflected by the dissatisfaction value, even if the agent could not compare its whole set of items with the allocation of the other agent and thus would be seen as envy-free.
Therefore, the dissatisfaction value allows a more nuanced representation of comparative unhappiness than the strict definition of envy.

\medskip
In this paper we want to find an allocation with maximal total satisfaction over all agents (i.e., the sum of each agent's satisfaction), which is equivalent to minimizing total dissatisfaction. In terms of social welfare induced by an allocation, maximizing total satisfaction reflects the idea of maximizing \textit{utilitarian social welfare}, hence aiming at an efficient allocation. We point out that in our companion paper~\cite{minmax}, the goal is to maximize \textit{egalitarian social welfare}, where the aim is to make the least satisfied agent as happy as possible (see also Section~\ref{sec:contr}).
Our research agenda focuses on the characterization of the boundary between \NP-hard and polynomially solvable cases of the problem depending on structural properties of the given preference graph.

The computational complexity involved in finding allocations of indivisible items that aim at maximizing social welfare has been studied for different notions of social welfare and different types of preference representation (see, e.g., \cite{bouveret-survey, Brandt}).
In this respect, besides utilitarian social welfare and egalitarian social welfare, also Nash social welfare has been in the focus of research (see, e.g.,~\cite{procaccia-Nash}). However, in contrast to our approach, typically the social welfare of an allocation is based on  individual utilities or scores from voting theory.  For the former, computational complexity results have been presented, for instance, for the representation of utilities in bundle form, $k$-additive form, or straight-line programs. For utilities expressed in
the bundle form, $\sf NP$-hardness results have been provided  for
utilitarian social welfare \cite{chevalyrekadditive}, egalitarian social welfare \cite{roosarothe},
and Nash product social welfare \cite{roosarothe, ulle}. In addition, already for the $1$-additive form (i.e., additive utilities), maximizing
egalitarian social welfare and Nash social welfare was shown to be \NP-hard in general \cite{Lipton,roosarothe}, and  maximizing
egalitarian social welfare is \NP-hard even for different scores from voting theory such  as the well-known Borda score~\cite{Baumeister2013}. Observe that in contrast to these works, in our paper the agents' preferences are expressed by a partial order over the items and a novel way of measuring the agents' satisfaction is considered.

\subsection{Contribution of the paper}\label{sec:contr}

At first we give a simple linear-time algorithm for general preference graphs with two agents.
Then we show that the problem is \NP-hard
for at least three agents, even for a rather simple class of preference graphs with no directed path of length two and no vertex of in-degree greater than $2$.
This hardness result for graphs of ``height'' two is complemented by a polynomial-time algorithm for preference graphs with width at most two for an arbitrary number of agents.
In our arguments we employ a close connection with the graph coloring problem (Section~\ref{sec:general}).

Preference statements often follow a hierarchical structure.
Thus, out-trees would be simple, but intuitive candidates for preference graphs.
The next conceptual level allows arbitrary directions on a tree.
Indeed, as a main result of the paper we present a linear-time algorithm for so-called polytrees, i.e., directed graphs whose underlying undirected graph is a tree (Section~\ref{sec:polytree}).
Note that this allows vertices with arbitrary in- and out-degree and implies preferences structures much more complex than simple directed out-trees.
The result also holds for a collection of polytrees, i.e., a polyforest.

Moving beyond tree structures, we then develop polynomial-time algorithms for classes of instances where the underlying undirected graph of the preference graph has treewidth at most two.
For $s,t$-series-parallel graphs, as well as for a directed version of cactus graphs, the problem can be solved for an arbitrary number of agents in polynomial time (Section~\ref{sec:treewidth2}).
Both of these graph classes generalize trees by allowing cycles, but without including complicated convolutions, which would be unlikely to appear in preference statements.

Some of the earlier results of this paper were included in the proceedings paper~\cite{adt2021}.

\subsection{Our related work}\label{sec:related}

In a companion paper~\cite{minmax}, we recently considered the same preference framework with a common preference graph but for a different objective function.
We studied a fairness criterion which minimizes the maximum dissatisfaction over all agents (or equivalently maximizes the minimum satisfaction).
In this setting, one wants to make the least satisfied agent as happy as possible without considering the total satisfaction.
It is shown in~\cite{minmax} that the same \PP\ vs.~\NP-complete complexity boundary between $k=2$ and $k\geq 3$ agents holds for that case.
Other polynomial results were derived for very simple graph structures, such as directed matchings and out-stars (see Section~\ref{sec:prelim} for formal definitions).
For out-trees, a polynomial result was given for a constant number of agents.
Furthermore, fixed parameter tractability results for a bounded number of agents were given for preference graphs that can be decomposed into path modules and independent set modules.
Compared to these findings, which employ algorithms quite different from the current paper, we show here that the minimization of total dissatisfaction permits polynomial algorithms for more general graphs.

Sticking to the utilitarian approach of minimizing the total dissatisfaction, it is also interesting to consider the case where each agent reports its own preference graph, as was done in~\cite{general}.
This general setting with individual preference graphs turned out to be much harder from a computational complexity perspective.

However, we could show that the problem
can be solved in polynomial time in each of the following cases~\cite{general}:
\begin{enumerate}[label*=(\arabic*)]
    \item if each preference graph is a disjoint union of paths,
    \item for two agents, if each preference graph is a disjoint union of out-stars, and
    \item for any constant number $k$ of agents if the underlying undirected graph of the union of the preference graphs has bounded treewidth $t$ (see~\cite[Theorem 16]{general} for details); in fact, for this case a fixed-parameter tractable algorithm was given with respect to $k+t$.
\end{enumerate}
A fortiori, these cases are also polynomial-time solvable for common preference graphs.
On the other hand, in the case of individual preference graphs the problem of minimizing the total dissatisfaction was shown in~\cite{general} to be \NP-complete in the case when each preference graph is an out-tree, as well as for two agents, even if the two sets of items are the same and no preference graph contains a two-edge directed path.
As a consequence of the results of the present paper, both of these intractable cases turn out to be solvable in polynomial time in the case of a common preference graph.

\section{Preliminaries}\label{sec:prelim}

For the sake of consistency, definitions and notation are mostly taken from~\cite{general}.
The description coincides with the framework in the companion paper~\cite{minmax}.

We consider directed and undirected finite graphs without loops or multiple edges.
For brevity, we often say \textit{graph} when referring to a directed graph.
A \emph{weakly connected component} of a directed graph $G$ is the subgraph of $G$ induced by the vertex set of a connected component of the underlying undirected graph of $G$.

Consider a directed graph $G=(V,A)$, where the $n:=|V|$ vertices correspond to $n$ items.
For $a=(u,v)\in A$, vertex $u$ is called the \emph{tail} of $a$ and vertex $v$ is called the \emph{head} of $a$.
The \textit{in-degree} of a vertex $u$ is the number of arcs in $A$ for which $u$ is the head and the \textit{out-degree} of $u$ is the number of arcs in $A$ for which $u$ is the tail.
The \textit{degree} of a vertex $u$ is the number of arcs in $A$ for which $u$ is either the head or the tail.
A vertex with in-degree $0$ is called a \textit{source}.
A vertex with out-degree $0$ is called a \textit{sink}.

A sequence $p=(v_0,v_1,v_2,\ldots,v_\ell)$ with $\ell \ge 0$
and $(v_i,v_{i+1})\in A$ for each $i\in \{0,\ldots,\ell-1\}$ is called a \textit{walk} of length $\ell$ from $v_0$ to $v_\ell$; it is a \textit{path} (of length $\ell$ from $v_0$ to $v_\ell$) if all its vertices are pairwise distinct.
A walk from  $v_0$ to $v_\ell$ is \textit{closed} if $v_0 = v_\ell$.
A \textit{cycle} is a closed walk of positive length in which all vertices are pairwise distinct, except that $v_0 = v_\ell$.

A \textit{directed acyclic graph} is a directed graph with no cycle. Observe that in a directed acyclic graph there is always at least one source and at least one sink.
An  \textit{out-tree} is a directed acyclic graph $G=(V,A)$ with a dedicated vertex $r$ (called the \textit{root}) such that for each vertex $v\in V\setminus\{r\}$ there is exactly one path from $r$ to $v$.
An \textit{out-forest} is a disjoint union of out-trees.
A directed acyclic graph whose underlying undirected graph is a tree (resp., a forest) is called a {\em polytree} (resp., a {\em polyforest}).

\medskip
A directed acyclic graph $G=(V,A)$ induces a binary relation $\succ$ over $V$, defined by setting $u\succ v$ if and only if $u\neq v$ and there is a directed path from $u$ to $v$ in $G$; in such a case, we say that the agents prefer $u$ to $v$.
In particular, a directed acyclic graph induces a strict partial order over $V$.

For $v\in V$ let ${\it \pred}(v)$ denote the set of predecessors of $v$, i.e., the set of all vertices $u\neq v$ such that there is a path from $u$ to $v$ in $G$.
In addition, let ${\it \succc}(v) \subseteq V$ denote the set of successors of $v$ in graph $G$, i.e., the set of all vertices $u\neq v$ such that there is a path from $v$ to $u$ in $G$.
We denote by $\pred[v]$ the set $\pred(v) \cup \{v\}$; similarly we denote by $\text{succ}[v]$ the set $\text{succ}(v) \cup \{v\}$.
For $u,v \in V$ we say that item $u$ is \textit{dominated} by item $v$ if $v\in {\pred}[u]$.

For a directed acyclic graph $G = (V,A)$ we denote by $N^-(v)$ the set of all in-neighbors of $v \in V$, formally $N^-(v)=\{u\in V\colon (u,v)\in A\}$, and similarly by $N^+(v)$ the set of all out-neighbors of $v$, that is $N^+(v)=\{u\in V\colon (v,u)\in A\}$.

An \emph{antichain} in a directed acyclic graph $G$ is a set of vertices that are pairwise unreachable from each other.
The \emph{width} of $G$ is the maximum cardinality of an antichain in $G$.

\medskip
We are given a set of $k$ agents denoted by $K=\{1,\ldots,k\}$.
An \textit{allocation} $\pi$ is a function $K \rightarrow  2^V$ that assigns to the agents pairwise disjoint sets of items, i.e., for $i,j \in K$, $i\not=j$, we have $\pi(i)\cap\pi(j)=\emptyset$.
Every allocation also implies a \textit{labeling} $\lambda_\pi$ of items, which is a function $V \rightarrow K\cup\{\NULL\}$ assigning to every item $v \in V$ allocated to agent $i$ the label $\lambda_\pi(v)=i$ and to items $v$ not allocated to any agent the label $\lambda_\pi(v)=\NULL$.
For a set of items $U \subseteq V$ we denote by $\lambda_\pi(U) = \{\lambda_\pi(v)\in K\colon v \in U\}$ the set of agents receiving items from~$U$.

As described before, we measure the attractiveness of an allocation by counting the number of items that an agent \emph{does not receive} and for which it receives no other more preferred item.
Formally, for an allocation $\pi$  the \textit{dissatisfaction} $\diss_{\pi}(i)$ of agent $i$ is defined as the number of items in $G$ not dominated by any item in $\pi(i)$.
A \textit{dissatisfaction profile} is a $k$-tuple $(d_i\colon i\in K)$ with $d_i \in \nat_0$ for all $i\in K$ such that there is an allocation $\pi$ with  $\diss_{\pi}(i)=d_i$ for each $i\in K$.

If $\pi$ allocates a vertex $v$ to agent $i$, it does not change the dissatisfaction $\diss_{\pi}(i)$ if any vertices in $\text{succ}(v)$ are allocated to $i$ in addition to $v$.
Thus, we say that if $v \in \pi(i)$, then agent $i$ \textit{dominates} all vertices in $\text{succ}[v]$.
It is also convenient to define the \textit{satisfaction} $s_\pi (i)$ of agent $i$ with respect to allocation $\pi$
as the number of items in $V$ that are dominated by $i$.

Continuing the work of \cite{general} we consider the minimization of the
total dissatisfaction among the agents.

\medskip
\probdef[Question]{A set $K$ of agents, a set $V$ of items, a  directed  acyclic graph $G =(V,A)$, and an integer $d$.}{Is there an allocation $\pi$ of items to agents such that the total dissatisfaction $\sum_{i\in K} \diss_{\pi}(i)$ is at most $d$?}{\textsc{Min-Sum Dissatisfaction with a Common Preference Graph\\(\MINSUM)}}{}{}

Sometimes we refer to the directed acyclic graph $G$ representing the agents' preferences as the \textit{common graph}.
The corresponding optimization problem asks for an \textit{optimal allocation}, i.e., an allocation that minimizes the total dissatisfaction of all agents.
As pointed out before, minimizing the total dissatisfaction of all agents is equivalent to maximizing the total satisfaction.
We use the former formulation for reasons of consistency with previous works~\cite{adt2021,minmax,general}.

Throughout the paper, for our polynomial-time results we provide constructive proofs, i.e., whenever we prove that the decision problem  {\MINSUM} can be decided in polynomial time we actually determine an optimal allocation in polynomial time.

\begin{lemma}\label{th:kgreatern}
The {\MINSUM} problem for a common graph $G = (V,A)$ with $n$ vertices and any set $K$ of $k \ge n$ agents has a canonical optimal solution independent of $A$, obtained by assigning each vertex of $G$ to a different agent.
\end{lemma}

\begin{proof}
Consider any optimal allocation $\pi$ violating the property of the Lemma.
Then there is an agent $i \in K$ with more than one assigned vertex, and thus there must be an agent $j \in K$ with no vertex at all, since $k \geq n$.
Now we reassign an arbitrary vertex $v\in \pi(i)$ from $i$ to $j$.
This decreases the dissatisfaction of $j$ by $|\textrm{succ}[v]|$ and increase the dissatisfaction of $i$ by \emph{at most} $|\textrm{succ}[v]|$, since $i$ may still be assigned vertices dominating some vertices in $\textrm{succ}[v]$.
Hence, in total this reassignment does not increase the total dissatisfaction.
This exchange step can be performed until each vertex is assigned to a different agent.

\end{proof}

Following \cref{th:kgreatern}, we assume that $k<n$ for the rest of the paper.

\medskip
\begin{observation}\label{obs:disconnected}
Let $\{G_1,\ldots, G_m\}$ be a set of pairwise vertex-disjoint graphs.
If {\MINSUM} can be solved in polynomial (resp., linear) time on each graph $G_j$, then {\MINSUM} can be solved in polynomial (resp., linear) time on $G=\bigcup_{j=1}^m G_j$.
\end{observation}

\begin{proof}
For each $j\in \{1,\ldots, m\}$, let $\pi_j$ be an optimal allocation for {\MINSUM} on $G_j$.
Then, we claim that assigning items $\bigcup_{j=1}^m \pi_j(i)$ to every agent $i \in K$ gives an optimal allocation $\pi$ for {\MINSUM} on $G$.
Since the graphs are pairwise vertex-disjoint, for each agent $i \in K$ we have that $\delta_\pi(i) = \sum_{j= 1}^m\delta_{\pi_j}(i)$.
Consequently, the total dissatisfaction of $\pi$ equals $\sum_{i \in K}\sum_{j= 1}^m\delta_{\pi_j}(i)= \sum_{j= 1}^m\sum_{i \in K}\delta_{\pi_j}(i)$.
Assume for a contradiction that there exists an allocation $\pi^*$ that has lower dissatisfaction than $\pi$.
Consider for all $j\in \{1,\ldots, m\}$ the allocation $\pi^*_j$ for {\MINSUM} on $G_j$ defined by $\pi^*_j(i) = \pi^*(i)\cap V(G_j)$ for all $i\in K$.
Then by our assumption $\sum_{j= 1}^m\sum_{i \in K}\delta_{\pi^*_j}(i)< \sum_{j= 1}^m\sum_{i \in K}\delta_{\pi_j}$.
Hence, there exists some $j\in \{1,\ldots, m\}$ such that $\sum_{i \in K}\delta_{\pi^*_j}(i)<\sum_{i \in K}\delta_{\pi_j}(i)$, contradicting the optimality of $\pi_j$ on $G_j$.
\end{proof}

\medskip
Many of our proofs employ the following general lemma on the solutions to {\MINSUM}.
(A similar lemma holds for the problem of minimizing the maximum dissatisfaction; see~\cite{minmax}.)

\begin{lemma}\label{structure-of-optimal-solutions}
When solving the {\MINSUM} problem for a common graph $G$ with $n$ vertices and any set $K$ of $k$ agents such that $k < n$, we may without loss of generality restrict our attention to allocations $\pi$ such that for each agent $i\in K$, the set of items allocated to agent $i$ forms a nonempty antichain in $G$.
\end{lemma}

\begin{proof}
Given an agent $i\in K$ such that $\pi(i)$ does not form an antichain in $G$, from $\pi(i)$ we can remove any item  that is dominated by some other item in $\pi(i)$ without changing the dissatisfaction of agent $i$.
Hence, by repeating this process if necessary, we end up with an allocation $\pi$ such that for each agent $i\in K$, the set of items  $\pi(i)$ is an antichain in $G$.

We now argue that we may assume that all these antichains are nonempty.
Let $\pi$ be an allocation such that for each agent $i\in K$, the set of items allocated to agent $i$ forms an antichain in $G$.
Assume that at least one of these antichains, say $\pi(j)$, for $j\in K$, is empty.
If not all vertices in $G$ are allocated, then we can assign any not yet allocated vertex of $G$ to agent $j$, thus strictly decreasing the total dissatisfaction.
Repeating this procedure if necessary, we may assume that all vertices in $G$ are allocated.
If at least one of these antichains, say $\pi(j)$, is still empty, then, since all the $n$ vertices in $G$ are allocated to at most $k-1\le n-1$ agents, there exists an agent $i\in K\setminus\{j\}$ such that $\{x,y\}\subseteq \pi(i)$ for two distinct vertices $x$ and $y$ of $G$.
Consider now the allocation that removes $y$ from $\pi(i)$ and assigns $\{y\}$ to agent $j$.
Then, the total dissatisfaction of this new allocation does not exceed the total dissatisfaction of the previous allocation.
Repeating this procedure if necessary, we eventually reach an allocation such that the set of items allocated to each agent forms a nonempty antichain in $G$.
\end{proof}

\section{General results for {\MINSUM}}\label{sec:general}

In~\cite{adt2021}, it was shown that {\MINSUM} is polynomial-time solvable if there are only two agents and becomes \NP-complete for $k=3$ agents, even if the graph has no directed path of length two and no vertex of in-degree larger than $2$.
This implies that the polynomial-time algorithm for $k=2$ is best possible regarding the number of agents, unless \textsf{P} = \NP{}.
The \NP-completeness proof for $k=3$ in~\cite{adt2021} is based on a fairly involved reduction from \textsc{$3$-Edge-Coloring}.

We now give shorter proofs of both results.
The proofs are based on the following lower bound on the optimal total dissatisfaction value and a characterization of allocations achieving this lower bound, which we use in several other proofs as well.

\begin{lemma}\label{lem:boundSum}
Given a preference graph $G=(V,A)$, for an arbitrary number $k$ of agents and any allocation $\pi$ the total dissatisfaction is bounded by
\begin{equation}\label{inequality-lower-bound}
\sum_{i \in K}\delta_\pi(i)\geq \sum_{v \in V}\max\{k-|\pred[v]|,0\}\,.
\end{equation}
\end{lemma}

\begin{proof}
We estimate the contribution of each vertex to the dissatisfaction separately.
There are $|\pred[v]|$ vertices that dominate vertex $v$.
Hence, since the sets of items assigned by $\pi$ to different agents are pairwise disjoint, at most $|\pred[v]|$ agents will be assigned an item that dominates vertex~$v$.
Consequently, at least $k-|\pred[v]|$ agents will not be assigned any item that dominates vertex~$v$.
In cases where the number of predecessors is large and the above number is negative the bound can clearly be improved to $\max\{k-|\pred[v]|,0\}$.
\end{proof}
Using topological sort (see, e.g.,~\cite{zbMATH07646025}), we can easily calculate $|\pred[v]|$ for all $v\in V$ in linear time.
In turn, $\sum_{v \in V}\max\{k-|\pred[v]|,0\}$ can also be calculated in linear time.

\medskip
Allocations for which the lower bound from~\eqref{inequality-lower-bound} is satisfied with equality can be characterized using the following definition.

\begin{definition}\label{def:good}
Given a preference graph $G=(V,A)$, any set $K$ of agents, and a set $S\subseteq V$, an allocation $\pi$ and its corresponding labeling $\lambda_\pi$ is said to be \emph{$S$-good} (with respect to $K$) if for each vertex $v\in S$ one of the following two conditions is satisfied:
\begin{itemize}
    \item $v$ is dominated by all agents in $K$, or
    \item the labels $\lambda_\pi(u)$ for $u \in \pred[v]$ are pairwise distinct and not $\NULL$.
\end{itemize}
A \emph{good} allocation is an allocation that is $V$-good.
\end{definition}

\begin{lemma}\label{good allocation result}
Given a preference graph $G=(V,A)$ for an arbitrary set $K$ of agents, an allocation $\pi$ achieves the lower bound from~\eqref{inequality-lower-bound} if and only if it is a good allocation.
\end{lemma}

\begin{proof}
Let $k=|K|$.
Assume that $\pi$ is good and consider an arbitrary vertex $v$ of $G$.
If $|\pred[v]|\le k$, then all vertices in $\pred[v]$ are assigned to different agents, since $\pi$ is $\{v\}$-good.
Thus, the contribution of $v$ to the total  dissatisfaction of $\pi$ is exactly $k-|\pred[v]| = \max\{k-|\pred[v]|,0\}$.
If $|\pred[v]| > k$, then each agent $i\in K$ receives an item that dominates $v$, since $\pi$ is $\{v\}$-good.
Therefore, the contribution of $v$ to the total dissatisfaction of $\pi$ is $0 = \max\{k-|\pred[v]|,0\}$.
It follows that the total dissatisfaction of $\pi$ meets the lower bound given by \cref{lem:boundSum}, and hence $\pi$ is optimal.

Conversely, assume an allocation $\pi$ meets the bound from \cref{lem:boundSum}.
Then it must hold that for every $v\in V$ we have $\max\{k-|\pred[v]|,0\}$ agents that do not dominate it.
If $\max\{k-|\pred[v]|,0\}=0$, then $v$ is dominated by all agents in $K$, implying $\pi$ is $\{v\}$-good.
If $\max\{k-|\pred[v]|,0\} \neq 0$, then there are $k-|\pred[v]|>0$ agents that do not dominate $v$, implying that $|\pred[v]|<k$ agents dominate $v$, hence $\lambda_\pi(u)$ for $u \in \pred[v]$ are pairwise distinct and not $\NULL$, implying $\pi$ is $\{v\}$-good.
\end{proof}

\Cref{lem:boundSum,good allocation result} immediately imply the following.

\begin{corollary}\label{good-implies-optimal}
For any preference graph $G$ and any number of agents, any good allocation is optimal for {\MINSUM}, that is, it minimizes the total dissatisfaction.
\end{corollary}

\Cref{good-implies-optimal} implies that in any class of instances in which a good allocation exists,
the decision problem of {\MINSUM} can be solved in polynomial time by evaluating the right-hand side expression in \eqref{inequality-lower-bound} and comparing it to the target value $d$.
If the good allocation can be computed in polynomial time, as it is the case for all positive results in this paper, this constitutes also an optimal solution to the corresponding optimization problem.
Furthermore, the use of good allocations simplifies both proofs of the mentioned results from~\cite{adt2021}.
We first give a short proof that for two agents {\MINSUM} can be solved in linear time.

\begin{theorem}\label{th:2agents}
For $k=2$ and any preference graph $G = (V,A)$, {\MINSUM} can be solved in $\mathcal{O}(|V|+|A|)$ time.
\end{theorem}

\begin{proof}
Let $S$ be the set of sources of $G$ (these correspond to items that are not dominated by any other item).
Let $G'$ be the graph $G-S$ and let $S'$ be the set of sources of $G'$.
Assigning all vertices in $S$ to agent $1$ and all vertices in $S'$ to agent $2$, we obtain a good allocation, which is an optimal allocation for {\MINSUM} by \cref{good-implies-optimal}.
\end{proof}

Next, we show that {\MINSUM} is \NP-hard for any constant number $k\ge 3$ of agents, even on graphs with no directed path of length two, which we call one-way directed bipartite graphs.

Formally, a \emph{one-way directed bipartite graph} is a graph $G=(X \cup Y, A)$ whose vertex set is partitioned into two sets $X$ and $Y$ such that, for each arc $(x,y)\in A$, it holds that $x \in X$ and $y \in Y$.
Given a one-way directed bipartite graph $G = (X \cup Y, A)$ we denote by $G_X$ the \emph{$X$-graph of $G$}, that is, the undirected graph with vertex set $X$ and an edge $\{x,x'\}$ between two distinct vertices $x,x' \in X$ if and only if there exists a vertex $y \in Y$ and arcs $(x,y), (x',y) \in A$.

Using the concept of $X$-graphs we give the following partial characterization of when a one-way directed bipartite graph admits a good allocation, valid if the maximum in-degree of all vertices in $G$, denoted by $\Delta^{-}(G)$, is bounded in terms of the number of agents.

\begin{lemma}\label{lemma:bipartite-good-allocation-coloring}
For any one-way directed bipartite preference graph $G = (X \cup Y, A)$ and any number of agents $k$ such that $k>\Delta^{-}(G)$, the graph $G$ admits a good allocation if and only if its $X$-graph is $k$-colorable.
\end{lemma}

\begin{proof}
    Given a $k$-coloring of the graph $G_X$ with colors $\{1,\ldots,k\}$ we construct a good allocation for $G$ as follows:
\begin{itemize}
    \item Vertices in $X$ colored $i$ are given to agent $i$.
    \item Vertices in $Y$ have less than $k$ parents, each assigned to a different agent since they are assigned based on a $k$-coloring of $G_X$.
    For each $y \in Y$, assign $y$ to an agent that is different from all the agents to which the parents of $y$ are assigned.
\end{itemize}
This assignment of vertices to agents is good, since for vertices in $X$ any allocation is good and for any vertex $y \in Y$ the parents of $y$ are assigned to pairwise distinct agents based on a $k$-coloring; since the in-degree of $G$ is less than $k$, vertex $y$ can be assigned to an agent distinct from the agents its parents are assigned to.

\medskip
For the converse direction, assume that $G$ admits a good allocation $\pi$ with labeling $\lambda_{\pi}$.
Since $\Delta^{-}(G) < k$, it holds that for every $y \in Y$ the parents of $y$ are assigned to pairwise distinct agents.
This implies that the labeling $\lambda_{\pi}$ restricted to $X$ is a $k$-coloring of the undirected graph $G_X$, since there is an edge $\{x,x'\} \in E(G)$ if and only if there exists a vertex in $Y$ that has both $x$ and $x'$ as a parent.
\end{proof}

\Cref{fig:notgood} gives an example that illustrates the requirement of $k$-colorability in \cref{lemma:bipartite-good-allocation-coloring}.

\begin{figure}
    \centering
    \begin{tikzpicture}[every node/.style={draw,circle,inner sep=0.8pt,minimum size=2mm}]
    \node[draw=none] at (5.8,1){$X$};
    \node[draw=none] at (5.8,0){$Y$};
\node (A1) at (1,1) {$1$};
\node (A2) at (2,1) {$2$};
\node (A3) at (3,1) {$3$};
\node (A4) at (4,1) {$4$};
\node (B12) at (0,0) {$3$};
\node (B13) at (1,0) {$4$};
\node (B14) at (2,0) {$2$};
\node (B23) at (3,0) {$1$};
\node (B24) at (4,0) {$1$};
\node (B34) at (5,0) {$2$};
\draw[->] (A1) -- (B12);
\draw[->] (A1) -- (B13);
\draw[->] (A1) -- (B14);
\draw[->] (A2) -- (B12);
\draw[->] (A2) -- (B23);
\draw[->] (A2) -- (B24);
\draw[->] (A3) -- (B13);
\draw[->] (A3) -- (B23);
\draw[->] (A3) -- (B34);
\draw[->] (A4) -- (B14);
\draw[->] (A4) -- (B24);
\draw[->] (A4) -- (B34);
\end{tikzpicture}
    \caption{A one-way directed bipartite graph $G$ with a good allocation for $k\geq4$, for which no good allocation for $k=3$ exists.
    The corresponding $X$-graph is the $4$-vertex complete graph, which is not $3$-colorable.}
    \label{fig:notgood}
\end{figure}

Using the connection between good allocations and $k$-colorings given by \cref{lemma:bipartite-good-allocation-coloring} we are able to prove the following hardness result.

\begin{theorem}\label{th:3agents}
For each $k\ge 3$, {\MINSUM} is \NP-complete, even if the preference graph is a one-way directed bipartite graph with maximum in-degree~$2$.
\end{theorem}

\begin{proof}
Clearly, {\MINSUM} belongs to \NP{} for any number of agents and preference graphs.
Fix an integer $k\ge 3$.
To prove \NP-hardness, we show that determining whether a given preference graph $G$ has a good allocation is \NP-hard, even if $G$ is a one-way directed bipartite graph with maximum in-degree~$2$.
By \cref{good allocation result}, this implies the theorem.

We reduce from \textsc{\hbox{$k$-Colorability}}, which is the following problem: ``Given an undirected graph $H$, is $H$ $k$-colorable?''
This problem is well known to be \NP-hard (see, e.g.,~\cite{zbMATH03639144}).
Take an instance $H$ of \textsc{$k$-Colorability}.
Create a directed graph $G = (V,A)$ from $H$ by doing the following.
Insert all vertices in $H$ into $V$.
For each edge $e = \{u,v\}$ of $H$, add a new vertex $w_e$ and insert arcs $(u,w_e),(v,w_e)$ in $A$.
In words, subdivide every edge of $H$ and orient the edges of the so obtained undirected graph towards the new vertices.
Taking $X = V(H)$ and $Y = \{w_e\colon e\in E(H)\}$, we obtain that $G=(X \cup Y, A)$ is a one-way directed bipartite graph with maximum in-degree~$2$.
Furthermore, $H$ is the $X$-graph of $G$, i.e., $H = G_X$.
Take $G$ as the input to the problem:
Does $G$ have a good allocation for $k$ agents?

Since $\Delta^{-}(G) = 2 < 3 \leq k$, and $G$ is a one-way directed bipartite graph, \cref{lemma:bipartite-good-allocation-coloring} implies that $H$ is $k$-colorable if and only if $G$ has a good allocation.
This completes our reduction, implying that the problem whether a graph has a good allocation for $k$ agents is \NP-hard.
\end{proof}

It should be noted that the analogous hardness result for minimizing maximum dissatisfaction (\hspace{1sp}\cite[Theorem 5]{minmax}) employs a related argument based on  \textsc{$k$-Colorability} and a similar graph construction.

\medskip
We conclude this section by observing that, in contrast to the intractability result of \cref{th:3agents} for the {\MINSUM} problem on graphs with ``height'' two, there is an efficient algorithm for graphs of \emph{width} at most two.

\begin{restatable}{theorem}{thmuno}\label{th:width-two-MINMAX-MINSUM}
For any number $k$ of agents, the {\MINSUM} problem is solvable in polynomial time if the preference graph $G$ has width at most~$2$.
\end{restatable}

The constructive proof of this statement is very similar to the corresponding proof for the minimization of the maximum dissatisfaction given in \cite[Theorem 6]{minmax}, except the last part.
For sake of self-containment of the description we give a full proof in Appendix~\ref{sec:appendix}.

\section{A linear-time algorithm for polytrees}\label{sec:polytree}

While in general, even for $k=3$, {\MINSUM} cannot be decided in polynomial time unless \PP{} = \NP{}, restricting the preference graph $G$ to tree-like structures allows for positive results.
Recall that it follows from~\cite[Theorem 16]{general} that for any constant number $k$ of agents {\MINSUM} is fixed-parameter tractable with respect to  $k+t$ where $k$ is the number of agents and $t$ is the treewidth of the underlying undirected graph of the preference graph.
In the rest of the paper, we identify some cases when polynomial-time solvability holds also for the case of an arbitrary number of agents.
We start with the simplest case, which is a particular case of graphs with treewidth one.

\begin{observation}\label{th:same-tree-const-MINSUM}
For any number $k$ of agents, {\MINSUM} can be solved in linear time if the preference graph $G$ is an out-tree.
\end{observation}

\begin{sloppypar}
\begin{proof}
Let us assume for simplicity that the set of agents is \hbox{$K=\{0,1,\ldots,k-1\}$}.
We define the \emph{depth} of a vertex $v$ as the length of the unique maximal directed path in $G$ ending at~$v$.
Process the vertices in order based on their depth and assign all items $v\in V(G)$ at depth $i\in \{0,1,\ldots,k-1\}$ to agent $i$ until we run out of agents or vertices.
The resulting allocation is good (cf.\ \cref{def:good}), and hence optimal by \cref{good-implies-optimal}.
\end{proof}
\end{sloppypar}

We now present our main result, a linear-time algorithm for {\MINSUM} on polytrees, i.e., directed graphs imposed on a tree structure.
Different from out-trees where a level-oriented optimal solution for {\MINSUM} can be stated explicitly, this allows vertices with in-degree larger than one, which makes the problem significantly more challenging.

The main idea of the algorithm is to process the polytree branch after branch (a branch is a directed path starting from a source vertex) and assign sources and vertices with in-degree $1$ first.
To be more precise, we start with an arbitrary source vertex $s$, assign it to some agent, and process a branch starting from $s$ until we reach a vertex $v$ with in-degree greater than~1, meaning there is at least one other branch through $v$.
At this point, we want to repeat the process on that branch, starting from one of its sources, and resuming the allocation with the next agent that does not dominate $v$ yet. This way we ensure that, at the end, we will have maximized the number of agents dominating $v$ and, furthermore, that the predecessors of $v$ will be in order assigned to consecutive agents.

Throughout the algorithm we employ a label $q(v) \in K \cup \{\NULL\}$ for each vertex $v$, which, when $q(v)\in K$, denotes the agent that has to be considered first to be assigned some predecessor of $v$.

\begin{algorithm}[tbh]
	\KwIn{A polytree $G = (V,A)$, a set $K = \{0,1,\ldots, k-1\}$ of agents.}
	\KwOut{A good allocation $\pi$ for {\MINSUM} given $G$ and $K$.}
	\tcp{Initialization}
	Let $s\in V$ be a source vertex of $G$, $\pi$ be an empty allocation, and let $L$ be a list initialized to $L = (s)$.\label{initialization}
	
	Set $q(s)=k-1$ and set $q(u)= \NULL$ for every $u \in V\setminus \{s\}$. \label{initialization-p} 
	
	\tcp{Main procedure}
	\While{$L\neq \emptyset$\label{while-loop}}{
	Let $v$ be the first element in $L$.\label{lookup-first-element-in-L}
	
	\If{there exists an in-neighbor $u$ of $v$ such that $q(u)=\NULL$  \label{ifclause}}{

	Add $u$ to the beginning of $L$.\label{prepending-vertices-to-L}
	
	Set $q(u)=q(v)$. \label{setpparent}
	}
	\Else{	\label{elseline}
    Update $\pi$ by assigning $v$ to the agent $q(v)$.\label{updatepi}
		
	\ForEach{out-neighbor $u$ of $v$\label{add out neighbors}}{

	    \If{$q(u)= \NULL$\label{if u not in L}}{
	        Add $u$ to the end of $L$.\label{appending-vertices-to-L}
	    }
	    Set $q(u)=(q(v)+1) \bmod k$. \label{setpchildren}
	}
	
	Remove $v$ from $L$.\label{remove-from-L}
	}}
	\Return{$\pi$}

	\caption{{\MINSUM} on polytrees.}\label{MINSUM oriented trees}
\end{algorithm}

\begin{theorem}\label{thm:forests}
For any number $k$ of agents, {\MINSUM} can be solved in linear time if the preference graph $G$ is a polytree.
\end{theorem}

\begin{proof}
Let $G = (V,A)$ be a polytree.
For ease of notation we define within this proof $K = \{0,1,\ldots, k-1\}$ as the set of agents.
Let $\pi$ be the output of \cref{MINSUM oriented trees} given $G$ and $K$.
We prove that \cref{MINSUM oriented trees} returns a good allocation $\pi$ through several claims that concern each vertex $x\in V$.

\begin{claim}\label{notnullifonlist}
At any iteration of the algorithm, when entering the while loop,
$q(x) \in K$
if and only if $x\in L$ or was in $L$ in a previous iteration.
\end{claim}

\begin{cproof}
Note that $q(x) \in K$ if and only if $q(x)\neq \NULL$.
There are two possibilities for changing $q(x)$ after its initialization.
Thus, if $q(x) \in K$, either $q(x)$ was set to some $q(v)$ in \cref{setpparent}, which means that $x$ was added to $L$ in the preceding \cref{prepending-vertices-to-L}, or $q(x)$ was set in \cref{setpchildren}.
In that case, when $q(x)$ was set in \cref{setpchildren} for the first time during the execution of the algorithm, it was identified as $q(x)= \NULL$ in \cref{if u not in L} and thus added to $L$ in \cref{appending-vertices-to-L}.

On the other hand, recall that $q(s) =k-1$ holds for the first element ever added to $L$.
Every item $x$ added later to $L$ in \cref{prepending-vertices-to-L} or \cref{appending-vertices-to-L} receives in the subsequent line a value depending on the current first element $v$ in $L$.
In particular, if this first element $v$ has $q(v) \in K$, then also $q(x) \in K$ according to \cref{setpparent,setpchildren}.
Thus, by induction no element added to $L$ can ever have the $q$ value set to something that is not in $K$.
\end{cproof}

\begin{claim} \label{notnull}
Whenever the allocation $\pi$ is updated by assigning some vertex $x$ to $q(x)$, the value of $q(x)$ is an agent in $K$.
\end{claim}

\begin{cproof}
Note that the allocation $\pi$ is only updated in \cref{updatepi}, where, if some vertex $x$ is assigned to $q(x)$, then $x$ is a vertex in the list $L$ and due to \cref{notnullifonlist} we have that $q(x)\in K$.
\end{cproof}

\begin{claim} \label{adding-to-L}
Each vertex $x$ is added to the list $L$ at most once.
\end{claim}

\begin{cproof}
Observe that due to \cref{notnullifonlist}, \cref{setpparent} or \cref{setpchildren} never assign a \NULL{} value.
Only vertices $x$ with $q(x)=\NULL$ are candidates to be added to $L$, and whenever they are added, their
$q$ value is immediately set to a value in $K$  and is never set to  \NULL{} again.
Hence we conclude that each vertex can be added to $L$ at most once.
\end{cproof}

\begin{claim} \label{complete-allocation}
Vertex $x$ is assigned to some agent in $K$ exactly once.
\end{claim}

\begin{cproof}
We first establish that the while loop terminates. In each iteration of the while loop a vertex gets added to $L$ or a vertex gets deleted from $L$.
\cref{adding-to-L} thus ensures that the while loop must finish within $2|V|$ iterations.
Second, vertices get deleted from $L$ only after they get assigned to some agent in $K$ by \cref{notnull}.
Hence, taking into account \cref{adding-to-L}, to prove the statement it is enough to show that every $x\in V$ is added to the list $L$.
Consider the step at which a vertex $v$ is removed from $L$.
Then $v$ has no in-neighbor $u$ with $q(u) = \NULL$; otherwise, the condition in \cref{ifclause} would be satisfied.
By \cref{notnullifonlist}, this means that every in-neighbor of $v$ has been added to $L$ at some iteration.
Furthermore, every out-neighbor of $v$ is added to $L$ at \cref{add out neighbors}.
Following the fact that $G$ is weakly connected, every vertex in $V$ will be added to $L$ at some iteration, and thus be assigned to some agent as pointed out above.
\end{cproof}

For simplicity, we denote by $L_{i-1}$ the state of the list $L$ at the $i$-th iteration when entering the while loop, starting with $i = 1$ and $L_0=(s)$.
Also, we write $L_{\leq i} = \{ u \in V \colon u \in L_j$ for some $0 \leq j \leq i\}$.

\begin{claim} \label{listconnected}
At any iteration of the algorithm,
when entering the while loop we have that
\begin{enumerate}[label*=(\alph*)]
    \item \label{no-inn-onL}
    for every $x \in L$, every $y \in \pred(x) \cap L$ appears before $x$ in $L$, and
    \item \label{weakly-connected} vertices $x$ with $q(x)\in K$ form a weakly connected subgraph.
\end{enumerate}
\end{claim}

\begin{cproof}
We prove the claim by induction on the iteration $i$ of the while loop.
Recall by \cref{notnullifonlist} that $q(x) \in K$ if and only if $x\in L_{\le i-1}$.
When $i = 1$, we have $L_{\leq 0} = \{s\}$ and $q(s) = k-1$, and both \ref{no-inn-onL} and~\ref{weakly-connected} hold.
Suppose the claim holds for all iterations $i \leq \ell$, and consider iteration $\ell+1$.
In particular, we are given $L_\ell$ and are constructing $L_{\ell+1}$.
By the induction hypothesis, for every $x \in L_\ell$, every $y \in \pred(x) \cap L_\ell$ appears before $x$ in $L_\ell$ and the vertices $x$ with $q(x)\in K$ form a weakly connected subgraph.

Consider first the case when an in-neighbor $u$ of $v$ is added at \cref{prepending-vertices-to-L}.
Then, $L_{\ell+1}$ is obtained from $L_\ell$ by adding $u$ to the beginning and, clearly, \ref{weakly-connected} is preserved.
Suppose for a contradiction to~\ref{no-inn-onL}
that $u$ has a predecessor $z$ in $L_\ell$.
Notice that this implies the existence of a path $P_{z,u}$ from $z$ to $u$ in $G$ that does not contain $v$, since $v$ is a successor of $u$.
Since the subgraph of $G$ induced by $L_{\leq \ell}$ is weakly connected, there exists a path $P_{z,v}$ between $z$ and $v$ in the underlying undirected graph of $G$ such that $P_{z,v}$ does not contain $u$.
However, the union of the vertex sets of $P_{z,u}$ and $P_{z,v}$ in the underlying undirected graph of $G$ induces a subgraph that contains a cycle, contradicting the fact that $G$ is a polytree.
Of course, \ref{no-inn-onL} holds for every other vertex in $L_{\ell+1}$, by the induction hypothesis.

Suppose now that an out-neighbor $u$ of $v$ is added at \cref{appending-vertices-to-L}.
It suffices to show that conditions \ref{no-inn-onL} and~\ref{weakly-connected} are preserved by each such operation.
Let $u_1,\ldots, u_p$ be the ordering of the out-neighbors $u$ of $v$ such that $q(u) = \NULL$ in which the algorithm adds them to $L$ in \cref{appending-vertices-to-L}, and, for all $j\in \{1,\ldots, p\}$, let $L_{\ell}^j$ denote the list $L$ right before $u_j$ is added.
Clearly, adding $u_j$ to $L$ preserves~\ref{weakly-connected}, and furthermore, $u_j$ is added at the end of the previous list.
Suppose for a contradiction that $u_j$ is a predecessor of some vertex $z$ in the current list $L_{\ell}^j$.
Since $u_j$ is an out-neighbor of $v$ and $v$ is a predecessor of $u_j$, there exists a path $P_{u_j,z}$ between $u_j$ and $z$ that does not contain $v$.
Again, as the subgraph of $G$ induced by $L_{\leq \ell}$ is weakly connected, there exists a path $P_{v,z}$ between $v$ and $z$ in the underlying graph of $G$ that does not contain $u$.
However, the union of the vertex sets of $P_{u_j,z}$ and $P_{v,z}$ induce a subgraph of the underlying graph of $G$ containing a cycle, contradicting the fact that $G$ is a polytree.
We conclude that $L_{\ell+1}$ satisfies \ref{no-inn-onL} for every vertex $x$ in $L_{\ell+1}$.
\end{cproof}

\begin{claim} \label{predecessors}
At any time when $x$ is assigned to an agent, every predecessor of $x$ has already been assigned to some agent.
\end{claim}

\begin{cproof}
Note that by \cref{complete-allocation} every $x \in V$ gets assigned in some iteration of the while loop.
Note also that if for some vertex $y\in V$ we have $y \in L_{\leq i}$ and $y \not\in L_i$, and because vertices are removed from the list only in \cref{remove-from-L}, after they have been assigned, this implies that $y$ has been assigned to some agent in an iteration previous to the $i$-th.
Thus, it suffices to show that if $x\in V$ is assigned at iteration $i$, then for every $z\in \pred(x)$ we have $z\in L_{\leq i}$ and $z\notin L_{i}$.

The proof is by induction on the iteration $i$ of the while loop.
The statement trivially holds for any source vertex, since it has no predecessors, in particular also for $L_{\leq 0}=\{s\}$, which settles the case $i=1$.
For $\ell\ge 2$, assuming the statement holds for all $i \leq \ell-1$, we show the statement for $i=\ell$.

Assume $x$ is the first vertex in list $L_\ell$ and that $x$ is not a source vertex.
Note that by \cref{ifclause} for every $y\in N^-(x)$ we have $q(y)\neq\NULL$ and $q(y)$ must have been changed in \cref{setpparent} or \cref{setpchildren} to some value in $K$.
By \cref{notnullifonlist} for each such $y$ it holds that $y \in L_{\leq \ell}$ and by \cref{listconnected}\eqref{no-inn-onL} $y \not\in L_\ell$.
We can therefore conclude that $y$ has been assigned to an agent in some iteration previous to $\ell$.
By the induction hypothesis the statement now holds for all predecessors of $y$.
Since the induction hypothesis states that for every $z \in \pred(y)$ we have $z \in L_{\leq \ell-1}$ and  $z \not\in L_{\ell-1}$,
and since $\pred(x) = \bigcup_{y \in N^-(x)} \pred[y]$
it follows that for every $z' \in \pred(x)$ we have $z' \in L_{\leq \ell}$ and  $z' \not \in L_{\ell}$, as claimed.
\end{cproof}

In the next claim we establish some properties of the allocation $\pi$ at that time vertex $x$ is assigned by the algorithm, which will imply that $\pi$ is $\{x\}$-good at that time and that the set of agents who get assigned items in $\pred[x]$ forms an interval.
The following notation will be used.
Given two integers $a$ and $b$, we define $[a,b] = \{j\in \mathbb{Z}\colon a\le j\le b\}$ and
\[[a,b]_k=\begin{cases}
[a \bmod k,b \bmod k] & \textnormal{if } (a \bmod k)  \leq (b \bmod k)\,,\\
[a \bmod k,k-1] \cup [0,b\bmod k] & \textnormal{otherwise.}
\end{cases}
\]
Note that the cardinality of the set $[a,b]_k$ equals $(b -a) \bmod k+1$.

\begin{claim}\label{morethangood}
For every vertex $x$, when $x$ is assigned we have that either
\begin{enumerate}[label=(\alph*)]
    \item \label{type1good}
    $\{i\in K\colon\pi(i) \cap \pred[x] \neq \emptyset\}
= K$ (i.e., every agent dominates $x$), or
    \item \label{type2good}
$\{i\in K\colon\pi(i) \cap \pred[x] \neq \emptyset\}
= [q(x)-|\pred(x)|,q(x)]_k$ (i.e., the set of agents dominating $x$ forms an interval of size $|\pred [x]|$); consequently, all items in $\pred[x]$ get assigned to pairwise distinct agents.
\end{enumerate}
\end{claim}

\begin{cproof}
If every agent in $K$ dominates $x$, then~\ref{type1good} holds.
We can thus assume that not all agents dominate $x$.
We prove that in this case~\ref{type2good} holds for $x$ by induction on the length of a longest path starting at a source vertex and ending at $x$.
It is clear that the statement holds for all source vertices (in particular for~$s$), as each such vertex $x$ satisfies $|\pred(x)| = 0$ and is dominated by the only agent from $\{q(x)\} = [q(x)-0,q(x)]_k$.

Assume $x$ is not a source vertex and let $N^-(x)=\{u_1,\ldots,u_\ell\}$ be the in-neighbors of $x$ in the order they are assigned.
Since not all agents dominate $x$, and since by \cref{predecessors} all predecessors of $x$ have been assigned, it follows that no vertex in $N^-(x)$ is dominated by all agents, either.
Hence, by induction hypothesis, \ref{type2good} holds for every in-neighbor of $u$, that is, $\{i\in K\colon\pi(i) \cap \pred[u] \neq \emptyset\} = [q(u)-|\pred(u)|,q(u)]_k$ for all $u \in N^-(x)$.

In particular, for $j\in \{1,\ldots,\ell-1\}$ and two consecutive in-neighbors $u_j$ and $u_{j+1}$ of $x$ we have that
\begin{align}
\{i\in K\colon\pi(i) \cap \pred[u_j] \neq \emptyset\} &= [q(u_j)-|\pred(u_j)|,q(u_j)]_k\,,\nonumber\\
\{i\in K\colon\pi(i) \cap \pred[u_{j+1}] \neq \emptyset\} &= [q(u_{j+1})-|\pred(u_{j+1})|,q(u_{j+1})]_k\,.\label{eq:intervals}
\end{align}
We establish consecutiveness of intervals (i.e.,  condition \ref{type2good} for vertex $x$) by showing that we also have
\begin{equation}\label{eq:consecutive}
q(u_j)+1\equiv q(u_{j+1})-|\pred(u_{j+1})| \pmod k\,.
\end{equation}
To this end, consider the iteration of the while loop when $u_j$ is assigned.
We have that $q(x)=(q(u_j)+1) \bmod k$, and in a later iteration, when $u_{j+1}$ is considered we have that $q(u_{j+1})=q(x)=(q(u_j)+1) \bmod k$ and $u_{j+1}$ is added to the beginning of the list $L$.
We claim that when this happens, no in-neighbors of $u_{j+1}$ are/were in the list $L$ yet.
Indeed, if $u_{j+1}$ has an in-neighbor $z$ that was already added to $L$ (and possibly already removed), then \cref{listconnected}\ref{weakly-connected} proves the existence of a path $P_{x,z}$ between $x$ and $z$ not containing $u_{j+1}$ in the underlying graph $\widehat{G}$ of $G$, which together with the edges $\{x,u_{j+1}\}$ and $\{u_{j+1},z\}$ implies the existence of a cycle in $\widehat{G}$, contradicting the fact that $G$ is a polytree.

If $u_{j+1}$ is a source vertex, then $u_{j+1}$ is the next assigned vertex in the algorithm and gets assigned to $q(u_{j+1})=(q(u_j)+1) \bmod k$, which establishes~\eqref{eq:consecutive}.

If $u_{j+1}$ is not a source vertex, then the algorithm proceeds up the tree until it reaches a source vertex.
Let $\{u_{j+1}=u'_1,\dots,u'_r\}$ be vertices visited in this order from $u_{j+1}$ to $u'_r$, where $u'_r$ is a source vertex.
By \cref{setpparent} all these vertices get the $q$ value set to $(q(u_j)+1) \bmod k$, in each iteration the considered vertex is added to the front of list $L$, and no vertex is assigned in this process.
In the iteration that follows, $u'_r$ is assigned to agent $(q(u_j)+1) \bmod k$.
On the other hand, since $u_r'$ is the first vertex from the set $\pred[u_{j+1}]$ that is assigned to an agent, \eqref{eq:intervals} implies that  $u_r'$ is assigned to agent $(q(u_{j+1})-|\pred(u_{j+1})|)\bmod k$.
Therefore, the equality
$q(u_j)+1\equiv q(u_{j+1})-|\pred(u_{j+1})| \pmod k$ holds, establishing~\eqref{eq:consecutive}.
This proves that, among the predecessors of $x$, the intervals are consecutive.

We are now ready to complete the induction step, that is, to establish the equality \[\{i\in K\colon\pi(i) \cap \pred[x] \neq \emptyset\}
= [q(x)-|\pred(x)|,q(x)]_k\,.\]
Since among the predecessors of $x$, the intervals are consecutive and not all agents in $K$ dominate $x$, those intervals are pairwise disjoint.
Hence, the set of agents dominating $x$ is exactly the set $[q(u_{1})-|\pred(u_{1})|,q(u_\ell)]_k \cup \{q(x)\}$.
As $u_\ell$ is the last in-neighbor of $x$ that will be assigned to some agent, it follows that when $x$ is assigned, $q(x)=(q(u_\ell)+1) \bmod k$.
By~\eqref{eq:consecutive}, we have $q(u_j)=(q(u_{j+1})-|\pred(u_{j+1})|-1) \bmod k$ for all $j\in \{1,\ldots, \ell-1\}$.
This implies that
\begin{align*}
q(u_1)&=\Big(q(u_{\ell})-\sum_{j = 2}^\ell(|\pred(u_{j})|+1)\Big) \bmod k \\
&=\Big(q(u_{\ell})-\sum_{j= 2}^\ell(|\pred[u_{j}]|)\Big) \bmod k\\
&=\Big(q(u_{\ell})-\sum_{j=1}^\ell |\pred[u_j]|+|\pred[u_1]|\Big) \ \bmod k\\
&=\big(q(u_{\ell})-|\pred(x)|+|\pred[u_1]| \big) \bmod k\,,
\end{align*}
and hence the set of agents dominating $x$ is exactly
\begin{align*}
\phantom{=}&\big[q(u_{1})-|\pred(u_{1})|\,,\,q(u_\ell)\big]_k \cup \big\{q(x)\big\}
=\big[q(u_{1})-|\pred(u_{1})|\,,\,q(u_\ell)+1\big]_k\\
=&\big[q(u_{\ell})-|\pred(x)|+|\pred[u_1]|-|\pred(u_{1})|\,,\,q(u_\ell)+1\big]_k\\
=&\big[q(u_{\ell})+1-|\pred(x)|\,,\,q(u_\ell)+1\big]_k =\big[q(x)-|\pred(x)|\,,\,q(x)\big]_k\,.
\end{align*}
This completes the induction step and with it the proof of the claim.
\end{cproof}

Let us now explain how the above claims imply that the algorithm is correct by showing that the output $\pi$ is indeed a good allocation.
First, note that \cref{complete-allocation} implies that the resulting mapping $\pi$ is a complete allocation.
\Cref{morethangood} implies that when a vertex $x$ is assigned, the allocation $\pi$ at that time is $\{x\}$-good (since the cardinality of the set $[q(x)+1-|\pred(x)|,q(x)+1]_k$ equals $|\pred(x)|+1 = |\pred[x]|$).
Using also Claims \ref{complete-allocation} and \ref{predecessors}, we infer that the property of $\pi$ being $\{x\}$-good is maintained also after the vertex $x$ has been assigned.
This implies that the final allocation $\pi$ is good and hence solves {\MINSUM} according to \cref{good-implies-optimal}, as claimed.

We now show that \cref{MINSUM oriented trees} can be implemented to run in linear time.

\begin{claim}\label{lem:linear-time}
\Cref{MINSUM oriented trees} admits a linear-time implementation.
\end{claim}

\begin{cproof}
We may assume $|V|\geq k$; otherwise at most $|V|$ agents would be assigned a vertex, and we could set $k = |V|$ to ensure linear-time complexity.
We can assume that for each vertex $v\in V$ we are given the lists of its in- and out-neighbors.
Furthermore, to achieve a linear-time complexity, we keep the in-neighborhood lists sorted so that for each vertex $v$, the in-neighbors of $v$ such that the $q$ value is $\NULL$ appear before all the other in-neighbors of $v$.
Every time that a vertex $u$ with $q(u) = \NULL$ gets its $q$ value set to a value in $K$ (in \cref{setpchildren,setpparent}), we iterate through all its out-neighbors $v$ and for each such vertex $v$ we move $u$ to the end of its list of in-neighbors.
With appropriate pointer structures, for a fixed vertex $u$, all these updates can be carried out in time $\mathcal{O}(|N^+(u)|)$, resulting in a total update time of $\mathcal{O}(|V|)$.

Initialization takes $\mathcal{O}(|V|)$ since we set the $q$ value for each vertex and the $\pi$ value for each $k$.
Since, by \cref{adding-to-L}, each vertex is added to the list $L$ at most once, the assumption on the in-neighbor lists assures that the overall time complexity of checking the {if} statement in \cref{ifclause}
is $\mathcal{O}(|V|)$.
\Cref{setpparent,prepending-vertices-to-L} can be done in $\mathcal{O}(1)$ time each.
By \cref{complete-allocation} the {else} statement happens $\mathcal{O}(|V|)$ times, since each vertex is assigned and deleted once.
\Cref{updatepi,remove-from-L} can be done in $\mathcal{O}(1)$ time.
We evaluate the {foreach} statement in \cref{add out neighbors} independently of the {while} loop.
Since each $v\in V$ is assigned exactly once, it follows that in \cref{add out neighbors} the algorithm considers each outgoing edge at most once, resulting in overall time complexity of $\mathcal{O}(|V|)$.
\Crefrange{if u not in L}{setpchildren} can be done in $\mathcal{O}(1)$ time.
We therefore conclude that the total running time of $\mathcal{O}(|V|)$ is achieved.
\end{cproof}

This concludes the proof of \cref{thm:forests}.\qedhere
\end{proof}

Since out-trees are a special case of polytrees, Theorem~\ref{thm:forests} also implies Observation~\ref{th:same-tree-const-MINSUM}.
Moreover, by \cref{obs:disconnected} the result also holds for polyforests.

\section{Polynomially solvable cases on graphs with\\ treewidth two}\label{sec:treewidth2}

In this section we give additional polynomial-time results for {\MINSUM} for an arbitrary number of agents if the preference graph $G$ is restricted to certain extensions of tree structures.
More specifically, we show that the problem is solvable in two classes of graphs for which the treewidth of the underlying undirected graph is at most two: $s,t$-series-parallel graphs and out-cactus graphs.
It is known that an undirected graph has treewidth at most two if and only if every biconnected component is a series-parallel undirected graph (see, e.g.,~\cite{MR1647486}).

\medskip
A graph $G=(V,A)$ is an \emph{$s,t$-series-parallel graph} if $G$ satisfies one of the following conditions:
\begin{itemize}
    \item $G$ consists solely of the edge $(s,t)$,
    \item $G$ is the series composition of $G_1$ and $G_2$ (denoted as $G = G_1 \circ G_2$), where $G_1$ is an $s_1,t_1$-series-parallel graph and $G_2$ is an $s_2, t_2$-series-parallel graph and $G$ is constructed by identifying the vertex $t_1$ with $s_2$ and setting $s=s_1, t = t_2$,
    \item $G$ is the parallel composition of $G_1$ and $G_2$ (denoted as $G = G_1 \parallel G_2$), where $G_1$ is an $s_1,t_1$-series-parallel graph and $G_2$ is an $s_2, t_2$-series-parallel graph and $G$ is constructed by identifying the vertex $s_1$ with $s_2$ and the vertex $t_1$ with $t_2$ and setting $s = s_1 = s_2, t = t_1 = t_2$.
\end{itemize}

\begin{theorem}\label{thm:series-parallel}
For any number $k$ of agents, {\MINSUM} can be solved in polynomial time if the preference graph $G$ is an $s,t$-series-parallel graph.
\end{theorem}

\begin{proof}
\begin{sloppypar}
For each $k = 1,\dots,|V|$ we denote the set of agents as
\hbox{$K=\{\bot, 0,1,\dots, k-2\}$} (where $K=\{\bot\}$ for $k=1$) and prove that there exists a good allocation $\pi^{(k)}$ of $G=(V,A)$ assigning items to all agents in $K$.
To simplify notation we denote by $\lambda^{(k)}$ the labeling of $\pi^{(k)}$ and in the remainder of the proof refer to labelings $\lambda^{(k)}$ instead of the corresponding allocations $\pi^{(k)}$.
The source $s$ of the $s,t$-series-parallel graph is always assigned to agent $\bot$.
We show the existence of all such $\lambda^{(k)}$ by induction on~$|V|$.
\end{sloppypar}
As the base case, note that for a two-vertex graph with an arc $(s,t)$ we set $\lambda^{(1)}(s) = \bot$, $\lambda^{(2)}(s) = \bot $ and $\lambda^{(2)}(t) = 0$.
Observe that the corresponding allocation satisfies the claim.

For the induction step, let $G$ be an $s,t$-series-parallel graph with at least three vertices and assume that the claim holds for all smaller $s,t$-series-parallel graphs.
We have $G = G_1 \circ G_2$ or $G = G_1 \parallel G_2$ where $G_1$ and $G_2$ are $s,t$-series-parallel graphs smaller than $G$.
For $i = 1,2$, let $G_i = (V_i, A_i)$, let $s_i$ and $t_i$ denote the unique source and sink of $G_i$, respectively, and let $\lambda_i^{(k_i)}$ for $k_i = 1,\dots,|V_i|$ denote the labeling of $G_i$ corresponding to a good allocation that exists by the induction hypothesis.

Assume first that $G = G_1 \parallel G_2$ with source $s=s_1=s_2$ and sink $t=t_1=t_2$.
We define $\lambda^{(k)}$ for each $k=1,\dots,|V|$.
For $i = 1,2$, let $k_i= \min\{k, |V_i|\}$.
For every $v \in V_1$, we set $\lambda^{(k)}(v) =\lambda^{(k_1)}_1(v)$.
For every $v \in V_2\setminus V_1$ such that $\lambda^{(k_2)}_2(v)\in K$, we set $\lambda^{(k)}(v) = (\lambda^{(k_2)}_2(v)+k_1-1) \mod (k-1)$.

\setcounter{claim}{0}

\begin{claim}
For each $k=1,\dots,|V|$ and $v \in V$ it holds that $\lambda^{(k)}$ for $G = G_1 \parallel G_2$ is $\{v\}$-good.
\end{claim}

\begin{cproof}
For each $v \in V \setminus \{t\}$ this follows directly by the properties of $\lambda_1^{(k_1)}$ and $\lambda_2^{(k_2)}$.
Hence we only have to show $\{t\}$-goodness.
If $t$ is dominated by all agents $K=\{\bot, 0, 1, \dots, k-2\}$, $\{t\}$-goodness holds by definition.
Otherwise, for each $i = 1,2$, we must have $|V_i|<k$ and thus $k_i = |V_i|$, implying that the labels $\lambda^{(k_i)}_i(u)$ are pairwise distinct for all $u \in \pred_{G_i}[t] = V_i$ and none of these labels is $\NULL$.
This also holds for $\lambda^{(k)}|_{V_2}$, since it is just a rotational shift of the labels in $\lambda^{(k_2)}_2$ (except for the source and the sink of $G_2$).
If $k_1 + k_2 \geq k+2$, then all agents in $K$ appear as labels of vertices in $V_1$ or $V_2$.
Otherwise, $k_1 + k_2 < k+1$.
In that case it holds that for each $j \in \{0,1,\ldots, k_2-2\}$ we have $j + k_1-1 < k-2$, hence $\lambda^{(k)}(V_1) \cap \lambda^{(k)}(V_2 \setminus V_1) = \emptyset$.
Hence, all vertices in $V$ are assigned to pairwise distinct agents in $K$, i.e., $\lambda^{(k)}$ is $\{t\}$-good.
\end{cproof}

Assume now that $G = G_1 \circ G_2$ with source $s = s_1$, sink $t=t_2$ and midpoint $c=t_1=s_2$.
We define $\lambda^{(k)}$ for each $k=1,\dots,|V|$ as follows.
Let $k_1 = \min\{k, |V_1|\}$ and $k_2 = \min\{k-k_1, |V_2| \}$.
For every $v \in V_1$, we set $\lambda^{(k)}(v) =\lambda^{(k_1)}_1(v)$.
Furthermore, if $k_2\ge 1$, then for every $v \in V_2\setminus V_1$ such that $\lambda^{(k_2)}_2(v)\in K$, we set $\lambda^{(k)}(v) = (\lambda^{(k_2)}_2(v)+k_1-1) \mod (k-1)$.

\begin{claim}
For each $k=1,\dots,|V|$ and $v \in V$ it holds that $\lambda^{(k)}$ as defined for $G = G_1 \circ G_2$ is $\{v\}$-good.
\end{claim}
\begin{cproof}
For each $v \in V_1$ this follows directly by the properties of $\lambda_1^{(k_1)}$.

Let $v \in V_2 \setminus V_1$.
If $k_2 = 0$, then $k_1 = k\le |V_1|$ and all vertices in $V_2$ are dominated by all agents due to the allocation of the items in $V_1$, and hence $\lambda^{(k)}$ is $\{v\}$-good.
Otherwise, it holds that $\lambda^{(k)}|_{V_2}$ is $\{v\}$-good in $G_2$ with respect to the set of agents $\lambda^{(k)}(V_2)\setminus\{\NULL\}$, since it is a rotational shift of $\lambda_2^{(k_2)}$ (except for the source of $G_2$).
Since $k_2 \leq k-k_1$, we have that  $\lambda^{(k)}(V_1) \cap \lambda^{(k)}(V_2 \setminus V_1) = \emptyset$.
Furthermore, since $k_2>0$, we have $k_1<k$ and hence $k_1 = |V_1|$.
As $\lambda_1^{(k_1)}$ is good, all items from $V_1$ are allocated to distinct agents in $\{\bot, 0, 1, \dots, k_1-2\}$. Hence, $\{v\}$-goodness of $\lambda^{(k)}$ follows from
the $\{v\}$-goodness of $\lambda_2^{(k_2)}$.
\end{cproof}
We conclude that $\lambda^{(k)}$ is a good allocation for $G$, and hence optimal (by \cref{good-implies-optimal}).
Therefore, {\MINSUM} can be solved in polynomial time on $s,t$-series-parallel graphs.
\end{proof}

Finally, we consider the graph class of out-cactus graphs.
In this context, a \emph{cactus graph} is an undirected graph such that each edge is contained in at most one cycle.
Equivalently, any two  cycles share at most one vertex.
A directed graph $G=(V,A)$ with a dedicated root vertex $r$ is called an \emph{out-cactus} if the underlying undirected graph is a cactus, $r$ is the only source, and for each cycle $C$ in $G$ it holds that there exists a unique source $s$ and a unique sink $t$, and all arcs entering the cycle $C$ in $G$ have $s$ as their tail.
It should be mentioned that out-cactus graphs are a superclass of \emph{galled trees} (where every vertex can be contained in at most one cycle) introduced in \cite{gus04} and widely studied in phylogenetics.
This connection was also pointed out in~\cite{nial18}.

For out-cactus preference graphs we can determine an optimal allocation by processing the graph in a top-down way and assigning vertices in any cycle using a topological order.

\begin{theorem}\label{thm:out-cactus}
For any number $k$ of agents, {\MINSUM} can be solved in polynomial time if the preference graph $G$ is an out-cactus.
\end{theorem}

\begin{proof}
We first decompose, in linear time, the edges of the underlying undirected graph $\widehat{G}$ of $G$ into cycles and edges not contained in any cycle.
For simplicity of the presentation of the algorithm, we also consider each arc $(u,v)$ of $G$ such that $\{u,v\}$ is not contained in any cycle of $\widehat{G}$, as a cycle of length one (with source $u$ and sink $v$).

We process the out-cactus $G$ in a top-down order starting from the source $r$, which we arbitrarily assign to some agent.
Whenever reaching the source $s$ of a new cycle $C$ with sink $t$,
the remaining vertices of this cycle are processed in the following manner.
Let $K'$ be the set of agents that were not assigned any vertex in $\pred[s]$ and $x_1, \dots, x_{\ell}$ be the sequence of vertices in $V(C)\setminus\{s\}$ in any topological order (in particular, $x_{\ell}=t$).
For $i \leq |K'|$ we assign $x_i$ to the $i$-th agent in $K'$ (type (i) assignment).
If $\ell > |K'|$, then each vertex in $x_{|K'|+1}, \dots, x_{\ell}$ is assigned to an arbitrary agent in $K'$ that was not assigned any of its predecessors, if such an agent exists (type (ii) assignment).

Let $v_1,\dots,v_n$ be the ordering of $V(G)$ in which the vertices are processed by the algorithm.
We claim that when vertex $v_i$, for some $i \in \{1,\dots,n\}$, is processed by the algorithm, the allocation is $\{v_i\}$-good.
We prove the claim by induction on $i$, and notice that the base case with $i=1$ holds trivially.
Consider the step when vertex $v = v_i$, for $1 < i \leq n$, is processed, and let $C$ be the corresponding cycle and $K'$ the set of agents as in the algorithm.
If $v$ is not assigned to any agent, then it is already dominated by every agent, in which case the allocation is $\{v\}$-good.
We may thus assume that $v$ is assigned to some agent.
Note that $v$ cannot have more than two parents.
First, suppose that $v$ has a unique parent $u$.
By the induction hypothesis, the current allocation is $\{u\}$-good.
Then vertex $v$ is assigned to some agent that is not yet dominating it, and we obtain a $\{v\}$-good allocation.
Now, suppose that $v$ has two parents $u_1$ and $u_2$.
Then $v$ is the sink of $C$.
By the induction hypothesis, the current allocation is $\{u_1,u_2\}$-good.
Note that if $u_1$ has not been assigned to any agent by the algorithm, then every agent dominates $u_1$ and thus $v$, which implies that the allocation is $\{v\}$-good.
A similar reasoning applies if $u_2$ has not been assigned.
If both $u_1$ and $u_2$ have been assigned with a type~(i) assignment, then all the vertices in the set $V(C) \setminus \{v\}$ have been assigned to different agents and, in particular, every vertex in $\pred(v)$ has been assigned to a different agent.
So, again, $v$ is assigned to some agent that is not yet dominating it, and we obtain a $\{v\}$-good allocation.
So we consider the case when at least one of $u_1$ and $u_2$ has been assigned with a type~(ii) assignment and assume without loss of generality that this is the case for $u_2$.
Let $S$ be the set of vertices of $C$ that have been processed before $u_2$ by the algorithm.
Since $u_2$ was not assigned by a type~(i) assignment, all agents in $K'$ were already assigned to vertices in $S$.
Furthermore, by the definition of $K'$, all agents in $K\setminus K'$ are assigned to vertices in $\pred[s]$ where $s$ is the source of $C$.
Since $\pred[s]\subseteq \pred[u_2]$, it follows that every agent in $K$ is assigned a vertex in $\pred(u_2) \cup S$.
Following the fact that $\pred[u_2] \cup S \subseteq \pred(v)$, we obtain that $v$ is dominated by every agent, and hence that the allocation is $\{v\}$-good.

We conclude that the algorithm produces a good allocation.
Hence, by \cref{good-implies-optimal}, the algorithm solves {\MINSUM} in polynomial time.
\end{proof}

The above result could be extended also to the case where every cycle of the out-cactus is replaced by an $s,t$-series-parallel graph.
However, we refrain from elaborating the technical details of the resulting algorithm.
Note that the class of out-cacti and the class of $s,t$-series-parallel graphs are incomparable, since out-cacti can have more than one sink and the class of out-cacti does not contain any $s,t$-series-parallel graph consisting of two vertices joined by at least~$3$ internally vertex-disjoint paths of length at least~$2$.

\section{Conclusions}
\label{sec:conc}

This paper studies the efficient allocation of $n$ indivisible items among $k$ agents with identical preferences.
Here, efficiency is expressed by the goal of maximizing utilitarian social welfare, i.e., maximizing the sum of satisfaction values or, equivalently, minimizing the sum of dissatisfaction values over all agents.
A companion paper~\cite{minmax} considered the fairness criterion of egalitarian social welfare by minimizing the maximum dissatisfaction among all agents.
Satisfaction and dissatisfaction result from a common preference graph shared between all agents.
In this acyclic digraph satisfaction of an agent counts the number of all vertices assigned to that agent plus all vertices that are dominated by these.
Dissatisfaction is the complement of satisfaction, i.e., the number of undominated vertices without the allocated vertices.

We study the computational complexity of the problem and derive a strong \NP-hardness result for $k\geq 3$ agents even for a highly restricted preference graph.
On the other hand, the problem can be solved in linear time for $k=2$ agents on arbitrary graphs and for an arbitrary number of agents on special graph classes, where the underlying undirected graph is a tree. In addition, for an arbitrary number of agents the problem is solvable in polynomial time when the underlying directed graph belongs to certain subclasses of graphs with treewidth two.
In particular, whereas for the class of out-trees the problem is known to be \NP-complete for agents with individual (non-identical) preference graphs~\cite{general},  {\MINSUM}, i.e., the problem for agents with a common preference graph,  turns out to be solvable in linear time. This also settles the computational complexity for the class of out-stars, for  which the computational complexity is still open when agents do not necessarily have a common preference graph.

However, also for agents with a common preference graph, there are several graph classes for which the complexity is still an open problem. Particular examples, when considering an arbitrary number of agents, are the classes of graphs of bounded width or bounded treewidth.

\subsubsection*{Acknowledgements.}

The authors are very grateful to the anonymous reviewers for their helpful comments, which helped a lot to improve the presentation of the paper.

The work of this paper was done in the framework of a bilateral project between University of Graz and University of Primorska, financed by the OeAD (SI 13/2023) and the Slovenian Research and Innovation Agency (BI-AT/23-24-009).
The authors acknowledge partial support of the Slovenian Research and Innovation Agency (I0-0035, research programs P1-0285 and P1-0404, and research projects J1-3003, J1-4008, J1-4084, J1-60012, N1-0210, and N1-0370), by the research program CogniCom (0013103) at the University of Primorska, and by the Field of Excellence ``COLIBRI'' at the University of Graz.

\subsubsection*{Declaration of interest:}
The authors have no conflicts of interest to declare that are relevant to the content of this article.

\begin{appendices}
        \renewcommand{\appendixname}{Appendix}
\section{Proof of Theorem~\ref{th:width-two-MINMAX-MINSUM}}\label{sec:appendix}

Let us first state some preliminaries.
A \emph{chain} in $G$ is a set $C$ of vertices such that for every two vertices $x,y\in C$, the graph $G$ contains either an $x,y$-path or a $y,x$-path.
A \emph{chain partition} of $G$ is a family $\mathcal{C}$ of vertex-disjoint chains in $G$ such that every vertex of $G$ belongs to precisely one chain in $\mathcal{C}$.
Due to Dilworth's theorem, the width of a graph $G$ equals the minimum number of chains in a chain partition of $G$~\cite{Dilworth}.
In order to efficiently determine  a minimum chain partition of $G$, Fulkerson's approach~\cite{MR0078334} can be used, solving a maximum matching problem in a bipartite graph with $2n$ vertices (see~\cite{MR545530}).
Applying the algorithm of Hopcroft and Karp~\cite{MR0337699}, this can be done in  $\mathcal{O}(n^{5/2})$ time.

In the proof of Theorem~\ref{th:width-two-MINMAX-MINSUM} below, we will also make use of the problem of computing a matching of cardinality $k$ in a bipartite graph such that the total weight of the matching edges is as small as possible.
This variant of the classical \textsc{Linear Sum Assignment} problem, which asks for a perfect matching with minimum total weight in a bipartite graph (see \cite[Sec.~6.2]{assign2012}), was  considered in \cite{Amico97}.

Omitting details we state the following lemma.

\begin{lemma}[folklore]\label{matching-k-min-sum}
Given a bipartite graph $G = (V,E)$ having a perfect matching, an edge weight function $w:E\to \mathbb{R}_+$, and an integer $k\le |V|/2$, in polynomial time we can compute a matching in $G$ with cardinality $k$ such that the total weight of the matching edges is as small as possible.
\end{lemma}

We are now ready to prove Theorem~\ref{th:width-two-MINMAX-MINSUM}.

\thmuno*

\begin{proof}
Let $G = (V,A)$ be a directed acyclic graph of width at most $2$.
Let us associate to $G$ an undirected bipartite graph $\widehat G = (\widehat V,E)$ defined as follows:
\begin{enumerate}
\item The vertex set of $\widehat G$ is $\widehat V = V\cup V'$ where $V' = \{v':v\in V\}$ is a set of new vertices.
\item There are two types of edges in $\widehat G$: edges of the form $\{x,y\}\subseteq V$, where $\{x,y\}$ is an antichain
in $G$ with size two, and edges of the form $\{v, v'\}$ where $v\in V$.
\end{enumerate}
Note that $\widehat G$ is bipartite.
Indeed, if $G$ has width one, then each of $V$ and $V'$ is an independent set in $\widehat G$ and their union is $\widehat V$.
If $G$ has width two, then by Dilworth's theorem there exists a chain partition $\{P,Q\}$ of $G$ with size two, and the sets $V(P)\cup \{v':v\in V(Q)\}$ and $V(Q)\cup \{v':v\in V(P)\}$ are independent sets in $\widehat G$ with union $\widehat V$.
Furthermore, the edges of the form $\{v, v'\}$ where $v\in V$ form a perfect matching in $\widehat G$.

We now assign nonnegative weights to the edges of $\widehat G$ as follows.
For each edge $e = \{x,y\}\subseteq V$ corresponding to an antichain in $G$ with size two, we define the weight to $e$ to be the dissatisfaction of an agent $i\in K$ under any allocation $\pi$ such that $\pi(i) = \{x,y\}$, that is, the total number of vertices in $G$ that are not reachable by a path from $\{x,y\}$.
Similarly, for each edge of the form $\{v, v'\}$ where $v\in V$, we define the weight to $e$ to be the dissatisfaction of an agent $i\in K$ under any allocation $\pi$ such that $\pi(i) = \{v\}$, that is, the total number of vertices in $G$ that are not reachable by a path from $v$.
Clearly, this weight function can be computed in polynomial time given $G$.

When solving the {\MINSUM} problem for $G$ and $K$, \cref{structure-of-optimal-solutions} implies that each agent $i\in K$ will be allocated either only one item or two items forming an antichain.
To each such allocation $\pi$ we can associate a matching $M_\pi$ in $\widehat{G}$ as follows.
For each agent $i\in K$ receiving only one item, say $\pi(i) = \{v\}$ for some $v\in V$, we include in $M_\pi$ the edge $\{v,v'\}$.
For each agent $i\in K$ receiving some antichain $\{x,y\}\subseteq V$ of size two, i.e.\ $\pi(i) = \{x,y\}$, we include in $M_\pi$ the edge $\{x,y\}$.
Since the sets of items assigned to different agents are pairwise disjoint, the obtained set $M_\pi$ is indeed a matching.
Furthermore, since for each of the $k$ agents, the set of items allocated to the agent forms a \emph{nonempty} antichain in $G$, we can associate to each agent $i\in K$ a unique edge $e_i\in M_{\pi}$; in particular, $M_{\pi}$ has cardinality $k$.
By construction, the weight of the edge $e_i$ equals to the dissatisfaction of agent $i$ with respect to the allocation $\pi$.
Note that the above procedure can be reversed: to any matching $M = \{e_i:i\in K\}$ of cardinality $k$ in $\widehat G$, we can associate an allocation $\pi_M$ of items in $G$ to the $k$ agents such that the dissatisfaction of agent $i$ under $\pi_M$ equals the weight of the edge $e_i$.

It follows that the {\MINSUM} problem for $G$ and $K$ can be reduced in polynomial time to the problem of computing a matching in $\widehat G$ with cardinality $k$ such that the total weight of the matching edges is as small as possible.
By \cref{matching-k-min-sum}, this problem can be solved in polynomial time.
\end{proof}
\end{appendices}

\bibliographystyle{plain}
\bibliography{refs}

\end{document}